%% file: width.tex
\documentclass[12pt]{article} 
\usepackage{makeidx}
\usepackage[utf8]{inputenc}
\usepackage[a4paper]{geometry} 
\usepackage{amstext, amssymb, amsmath, amsthm}
\usepackage{color, graphics}
\usepackage{url}
\usepackage{framed}
\usepackage{tikz}
\usepackage{float}
\usepackage{listings} 
\usepackage{algorithm2e}
\usepackage{cite}
\usepackage{xcolor}
\newtheorem{theorem}{Theorem}
\newtheorem{lemma}[theorem]{Lemma}

\newtheorem{definition}[theorem]{Definition}
\newtheorem{proposition}[theorem]{Proposition}

\tikzstyle{vertex}=[circle,draw=black,fill=white,minimum size=20pt,inner sep=0pt]
\tikzstyle{oedge} = [draw,thick,->]
\tikzstyle{edge} = [draw,thick]
\tikzstyle{selected edge} = [draw,line width=5pt,-,red!50]
\tikzstyle{weight} = [font=\small]

\def\O{\mathcal{O}}

\def\col{$3$-colorability}

\def\NPh{\textsf{NP}-hard}

\def\a{amalgam-width}
\def\I{\mathcal{I}}
\def\phi{\varphi}

\def\N{\mathbf{N}}
\def\TREE{\mathcal{T}}
\def\cl{\mbox{cl}}

\def\msolp{\textsf{MSO-DECIDE}}
\def\ind{\mbox{ind}}
\def\Qv{Q_{E(K(v))}}
\def\Qv{Q_v}
\def\Jv{{J(v)}}
\def\JvI{{J_1(v)}}
\def\JvII{{J_2(v)}}
\def\mc{\mbox{count}}

\def\ACC{\mbox{ACCEPT}}

\def\C{\mathcal{C}}
\def\amop{\oplus^{K, D}}
\def\onetwo{\{1, 2\}}

\newcommand{\executeiffilenewer}[3]{%
    \ifnum\pdfstrcmp{\pdffilemoddate{#1}}%
        {\pdffilemoddate{#2}}>0{\immediate\write18{#3}}
    \fi}

\author{Lukas Mach}

\begin{document}

\newcommand{\defined}[1]{#1}

\title{Amalgam width of matroids}
\date{}
\author{
	Luk\' a\v s Mach\footnote{DIMAP and Department of Computer Science, University of Warwick, Coventry, United Kingdom. E-mail: \texttt{lukas.mach@gmail.com}.} \thanks{This author has received funding from the European Research Council under the European Union's Seventh Framework Programme (FP7/2007-2013)/ERC grant agreement no.~259385 and from the student project GAUK no. 592412 when being a PhD student at Computer Science Institute, Faculty of Mathematics and Physics, Charles University, Prague, Czech~Republic.}, 
	Tom\' a\v s Toufar\footnote{Computer Science Institute, Faculty of Mathematics and Physics, Charles University, Prague, Czech~Republic. E-mail: \texttt{tomas.toufar@gmail.com}}
} 
\maketitle 

\begin{abstract} 
We introduce a new matroid width parameter based on the operation of \textit{matroid amalgamation}, which we call \textit{\a}. 
The parameter is linearly related to branch-width on finitely representable matroids, while still allowing the algorithmic application on non-representable matroids (which is not possible for branch-width). 
In particular, any property expressible in the monadic second order logic can be decided in linear time for matroids with bounded amalgam-width.
We also prove that the Tutte polynomial can be computed in polynomial time for matroids with bounded \a.
\end{abstract} 


\section{Introduction}

It is well known that many \NPh\ graph problems can be solved efficiently when restricted to trees or to graphs with bounded tree-width.
Research of this phenomenon culminated in proving a celebrated theorem of Courcelle \cite{C1}, which asserts that any graph property expressible in the monadic second order (MSO) logic can be decided in linear time for graphs of bounded tree-width. 
Such properties include, among many others, \col.
There are several other width parameters for graphs with similar computational properties, e.g., boolean-width \cite{BX} and clique-width \cite{CMR}.

In this work, we study matroids, which are combinatorial structures generalizing the notions of graphs and linear independence. 
Although the tree-width for matroids can be defined \cite{Htw}, a more natural width parameter for matroids is branch-width.
This is due to the fact that the branch-width of graphs can be introduced without referring to vertices, which are not explicitly available when working with (graphic) matroids. 
We postpone the formal definition of branch-width to Section \ref{not} and just note that the branch-width of a matroid or a graph is linearly related to its tree-width. 

It is natural to ask to what extent the above-mentioned algorithmic results for graphs have their counterparts for matroids.
However, most width parameters (including branch-width) do not allow corresponding extension for general matroids without additional restrictions. 
Although computing decompositions of nearly optimal width is usually still possible (see \cite{Ocert, Oapprox}), 
the picture becomes more complicated for deciding properties. 
Extensions to finitely representable matroids are feasible but significant obstacles emerge for non-representable matroids. 
This indicates a need for a width parameter reflecting the complex behavior of matroids that are not finitely representable. 


Let us be more specific with the description of the state of the art for matroids.
On the positive side, the analogue of Courcelle's theorem was proven by Hlin\v en\'y \cite{Hbw} in the following form: 
\begin{theorem}
\cite[Theorem 6.1]{Hbw}
Let $\mathbf{F}$ be a finite field, $\phi$ be a fixed MSO formula and $t \in \N$. 
Then there is a fixed parameter algorithm deciding $\phi$ on $\mathbf{F}$-represented matroids of branch-width bounded from above by $t$.
\end{theorem} 
However, as evidenced by several negative results, a full generalization of the above theorem to all matroids is not possible: 
Seymour \cite{Sey} has shown that there is no sub-exponential algorithm testing whether a matroid (given by an oracle) is representable over $\mbox{GF}(2)$. 
Note that being representable over $\mbox{GF}(2)$ is equivalent to the non-existence of $U_2^4$ minor, which can be expressed in MSO logic. 
This result generalizes for all finite fields and holds even when restricted to matroids of bounded branch-width.
This subsequently implies the intractability of deciding MSO properties on general matroids of bounded branch-width
even when restricted to matroids representable over rationals. 
	See \cite{Crap} for more details on matroid representability from the computational point of view. 
Besides MSO properties, first order properties for matroids have been studied from algorithmic point of view in \cite{GKO}.

Two width parameters were proposed to circumvent the restriction of tractability results to matroids representable over finite fields:
decomposition width \cite{K} 
and another width parameter based on $2$-sums of matroids \cite{S}.
The latter allows the input matroid to be split only along $2$-separations, 
making it of little use for $3$-connected matroids. 
On the other hand, 
though the first one can split the matroid along more complex separations,
it does not correspond to any natural ``gluing'' operation on matroids.
In this work, we present a matroid parameter, 
called \a, that has neither of these two disadvantages and
it still allows proving corresponding algorithmic results.
An input matroid can be split along complex separations 
and the parts of the decomposed matroid can be glued together using the so-called amalgamation \cite{O}, 
which is a well-established matroid operation.



\section{Notation}
\label{not}

We now introduce basic definitions and concepts further used in the paper. 
The reader is also referred to the monograph \cite{O} for a more detailed exposition on matroid theory. 

\textit{A matroid} $M$ is a tuple $(E, \I)$ where $\I \subseteq 2^E$. 
The set $E$ is called \textit{the ground set}, its elements are \textit{the elements of $M$}, and the sets in $\I$ are called \textit{independent sets}. 
The ground set of a particular matroid $M$ is denoted by $E(M)$.
The set $\I$ is required 
	(1) to contain the set $\emptyset$,
	(2) to be hereditary, i.e., for every $F \in \I, \hspace{1.5mm} \I$ must contain all subsets of $F$, 
	and (3) to satisfy \textit{the exchange axiom}: 
	if $F$ and $F'$ are independent sets satisfying $|F| < |F'|$, 
	then there exists $x \in F'$ such that $F \cup \{x\} \in \I$.
We often understand matroids as sets of elements equipped with the property of ``being independent''. 
If a set is not independent, we call it \textit{dependent}.
A minimal depedent set is called \textit{a circuit}.
The set of all circuits of the matroid, denoted by $\mathcal{C}(M)$, uniquely determines the matroid. 

Examples of matroids include \textit{graphic matroids} and \textit{vector matroids}.
The former are derived from graphs in the following way: 
their elements are edges and a set of edges is independent if it does not span a cycle. 
Vector matroids have vectors as their elements and a set of vectors is independent if the vectors in the set are linearly independent.
A matroid $M$ is called representable over a field $F$ if there exists a vector matroid over $F$ isomorphic to $M$. 
Finally, a matroid is binary if it is representable over the binary field and it is regular if it is representable over any field. 

\textit{The rank} of a set $F$, denoted by $r(F)$, is the size of the largest independent subset of $F$ 
(it can be inferred from the exchange axiom that all inclusion-wise maximal independent subsets of $F$ have the same size). 
If $F \subseteq E(M)$, then \textit{the closure operator} $\cl(F)$ is defined as $$\cl(F) := \big\{x : r(F \cup \{ x \}) = r(F) \big\}.$$
It can be shown that $r(\cl(F)) = r(F)$.
A set $F$ such that $\cl(F) = F$ is called \textit{a flat of $M$}.

If $F$ is a subset of $E(M)$, then $M \setminus F$ is the matroid obtained from $M$ by deleting the elements of $F$, i.e., the elements of $M \setminus F$ are those not contained in $F$ and a subset $F'$ of such elements is independent in the matroid $M \setminus F$ if and only if $F'$ is independent in $M$. 
The matroid $M / F$ which is obtained by \textit{contraction of $F$} is defined as follows: the elements of $M / F$ are the ones not contained in $F$ and a $F' \subseteq E(M) \setminus F$ is independent in $(M / F)$ iff $r(F \cup F') = r(F) + r(F')$. 
For $F \subseteq E$, we define \textit{the restriction} $M|F$ as $M \setminus (E \setminus F)$.
\textit{A loop} of $M$ is an element $e$ of $M$ with $r(\{e\}) = 0$. 
\textit{A separation} $(A, B)$ is a partition of the elements of $M$ into two disjoint sets $A$ and $B$ and a separation $(A, B)$ is \textit{a $k$-separation} if $r(A) + r(B) - r(M) \leq k - 1$. 

\textit{A branch-decomposition of a matroid $M = (E, \I)$} is a tree $T$, in which 
\begin{itemize} 
\item the leaves of $T$ are in one-to-one correspondence with the elements of $E$ and 
\item all inner nodes have degree three.
\end{itemize} 
Every edge $e$ of $T$ splits the tree into two subtrees so that the elements corresponding to the leaves of the respective subtrees form a partition $(E_1, E_2)$ of the ground set.
\textit{The width of an edge $e$ of $T$} is defined as $r(E_1) + r(E_2) - r(E) + 1$, where $E_1$ and $E_2$ are the edge sets corresponding to the leaves of the two components of $T \setminus e$.
Thus, the width of an edge $e$ is the smallest $k$ such that the induced partition $(E_1, E_2)$ is a $k$-separation of $M$. 
\textit{The width of the branch-decomposition $T$} is a maximum width of an edge $e \in T$.
Finally, \textit{the branch-width $bw(M)$ of a matroid} is defined as the minimum width of a branch-decomposition of $M$.

The question of constructing a branch decomposition of a small width was positively settled in \cite{Ocert, Oapprox} for general matroids (given by an oracle).
\begin{theorem} 
\cite[Corollary 7.2]{Oapprox} 
For each $k$, there is an $\O(n^4)$ algorithm constructing a decomposition of width at most $3k - 1$ or outputting a true statement that the matroid has branch-width at least $k + 1$.
\end{theorem}
Moreover, for matroids representable over a fixed finite field, an efficient algorithm for constructing a branch decomposition of optimal width is given in \cite{HOcon}. 

\begin{figure}[h]
\centering
\begin{minipage}{0.49\textwidth}
\centering

  \includegraphics[height=1.3in]{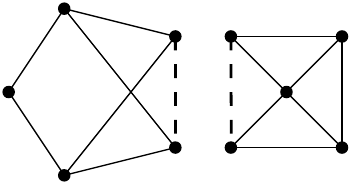}

\end{minipage}
\hspace{0.04\textwidth}
\begin{minipage}{0.38\textwidth}

  \centering
  \includegraphics[height=1.3in]{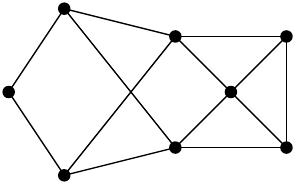}

\end{minipage}
\caption{The underlying graphs of matroids $M_1, M_2$ (with edges $p_1, p_2$ being dashed) and the underlying graph of the graphic matroid $M_1 \odot_{p_1,p_2} M_2$.}
\label{sum}
\end{figure}

Let $M_1$ and $M_2$ be two matroids satisfying $p_i \in E(M_i),$ for $i \in \{1, 2\}$.
Then the $2$-sum $M_1 \odot_{p_1,p_2} M_2$ is defined to be the matroid with the following set of circuits:
\begin{align*} 
\C =&\ \C(M_1 \setminus p_1) \cup\ \C(M_2 \setminus p_2)\ \cup \\
    &\ \{ (C_1 \setminus p_1) \cup (C_2 \setminus p_2) : p_i \in C_i \in \C(M_i) \text{ for } i \in \{1,2\} \}. 
\end{align*} 
An example of a $2$-sum of a pair of graphic matroids is can be found in Figure \ref{sum}.

We say that an algorithm runs in \textit{linear time} if it always finishes in $\O(n)$ steps, where $n$ is the length of the input in an appropriate encoding. 
Similarly, an algorithm runs in \textit{polynomial time} if it always finishes in $\O(n^k)$ steps, for $k \in \N$. 
When a part of the algorithm's input is given by an oracle (e.g., a rank-oracle specifying an input matroid), the time the oracle spent computing the answer is not counted towards the number of steps the main algorithm took -- only the time spent on constructing the input for the oracle and reading its output is accounted for in the overall runtime. 

\textit{A monadic second order formula $\psi$}, shortly an \textit{MSO formula}, for a matroid $M$ contains the basic logic connectives $\lor, \land, \neg, \Rightarrow$, quantifications over elements and subsets of $E(M)$ (which we refer to as element and set variables, respectively), the equality predicate, the predicate of containment of an element in a set, and, finally, the independence predicate determining whether a set of elements of the matroid is independent. 
The independence predicate encodes the input matroid.

Deciding MSO properties of matroids is \NPh\ in general, since, for example, the property that a graph is hamiltonian can be determined by deciding the following formula on the graphic matroid corresponding to the input graph: 
$$\exists H \exists e \big(\mbox{is\_circuit}(H) \land \mbox{is\_base}(H \setminus \{e\})\big),$$ where $H$ is a set variable, $e$ an element variable, and 
$\mbox{is\_circuit($\cdot$)}$ and $\mbox{is\_base($\cdot$)}$ are predicates testing the property of being a circuit and a base, respectively. 
These can be defined in MSO logic as follows:
\begin{align*} 
\mbox{is\_circuit}(H) & \equiv \big(\neg \ind(H)\big) \land \big(\forall e: (e \in H) \Rightarrow \ind(H \setminus \{e\}) \big), \\
\mbox{is\_base}(H) & \equiv \neg \big( \exists e: \ind(H \cup \{e\}) \big).
\end{align*}

\section{Matroid amalgams} 
\label{par}

In this section we define the operation of \textit{a generalized parallel connection}, which plays a key role in the definition of an amalgam decomposition. 
We begin by introducing matroid amalgams and modular flats. 

\begin{definition}
\label{amalgam}
Let $M_1$ and $M_2$ be two matroids.
Let $E = E(M_1) \cup E(M_2)$ and $T = E(M_1) \cap E(M_2)$.
Suppose that $M_1|T = M_2|T$. 
If $M$ is a matroid with the ground set $E$ such that $M|E_1 = M_1$ and $M|E_2 = M_2$, 
we say that $M$ is an \defined{amalgam} of $M_1$ and $M_2$.
\end{definition}

An amalgam of two matroids does not necessarily exist, even if the matroids coincide on the intersection of their ground sets. 
Our aim is to investigate a condition on matroids sufficient for the existence of an amalgam.
To do so, we introduce the notions of \textit{free amalgams} and \textit{proper amalgams.}

\begin{definition}
\label{free}
Let $M_0$ be an amalgam of $M_1$ and $M_2$. 
We say that $M_1$ is the \defined{free amalgam} of $M_1$ and $M_2$ if for every amalgam $M$ of $M_1$ and $M_2$ every set independent in $M$ is also independent in $M_0$.
\end{definition}

The definition of a more restrictive \textit{proper} amalgam is more involved.

\begin{definition}
\label{proper} 
Let $M_1$ and $M_2$ be two matroids with rank functions $r_1$ and $r_2$, respectively, and independent sets coinciding on $E_1 \cap E_2$. 
First, define functions $\eta$ and $\zeta$ on subsets of $E := E_1 \cup E_2$ as follows.
$$\eta(X) := r_1(X \cap E_1) + r_2(X \cap E_2) - r(X \cap T),$$
$$\zeta(X) := \min \{ \eta(Y) : Y \supseteq X \},$$
where $T := E_1 \cap E_2$ and $r$ is the rank function of the matroid $N := M_1 | T = M_2 | T$.
(Note that $\eta$ provides an upper bound on the rank of the set $X$ in a supposed amalgam of $M_1$ and $M_2$, while $\zeta$ is the least of  these upper bounds.) 
If $\zeta$ is submodular on $2^E$, we say that the matroid on $E_1 \cup E_2$ with $\zeta$ as its rank function is \defined{the proper amalgam} of $M_1$ and $M_2$.
\end{definition}

It can be verified that if the proper amalgam of two matroids exists then it is a free amalgam.  
The next lemma provides a necessary and sufficient condition for an amalgam to be the proper amalgam of two given matroids.
\begin{lemma} 
\label{necessary}
Let $M_1$ and $M_2$ be two matroids and $M$ one of their amalgams. 
$M$ is the proper amalgam of $M_1$ and $M_2$ if and only if it holds for every flat of $M$ that 
$$r(F) = r(F \cap E_1) + r(F \cap E_2) - r(F \cap T).$$
\end{lemma} 
However, Lemma \ref{necessary} says nothing about the existence of the proper amalgam of $M_1$ and $M_2$.
We next give a condition that guarantees it. 

\begin{definition}
\label{modular}
A flat $X = \cl(T)$ of a matroid $M$ is \defined{modular} if for any flat $Y$ of $M$ the following holds: 
$$r(X \cup Y) = r(X) + r(Y) - r(X \cap Y).$$
Furthermore, we say that $T$ is \defined{a modular semiflat} if $\cl(T)$ is a modular flat in $M$ and every element of $\cl_M(T)$ is either in $T$, a loop, or parallel to some other element of $T$.
\end{definition}

For example, the set of all elements, the set of all loops, 
	       	and any flat of rank one are modular flats. 
Each single-element set is a modular semiflat. 

\begin{theorem}
\cite{O}
\label{connection} 
Suppose that $M_1$ and $M_2$ are two matroids with a common restriction $N := M_1 | T = M_2 | T$, where $T = E(M_1) \cap E(M_2)$.
If $T$ is a modular semiflat in $M_1$, then the proper amalgam of $M_1$ and $M_2$ exists.
\end{theorem}
We are now ready to introduce the operation of a generalized parallel connection, which can be used to glue matroids.

If $M_1$ and $M_2$ satisfy the assumptions of Theorem \ref{connection}, then the resulting proper amalgam is called \textit{the generalized parallel connection} of $M_1$ and $M_2$ and denoted by $M_1~\oplus_N~M_2$, where $N := M_1 | (E(M_1) \cap E(M_2)).$
If we use $M_1~\oplus_N~M_2$ without specifying $N$ in advance, 
$N$ refers to the unique intersection of the two matroids.
The generalized parallel connection satisfies the following properties.

%

\begin{lemma} 
If the generalized parallel connection of matroids $M_1$ and $M_2$ exists, $\cl(E_2)$ is a modular semiflat in $M_1 \oplus_N M_2$. 
\end{lemma} 

\begin{lemma} 
\cite[p. 446]{O}
\label{closureissimple}
Let $M_1$ and $M_2$ be two matroids, $T = E(M_1) \cap E(M_2)$, $N$ the matroid $M_1 | T = M_2 | T$, and $M = M_1 \oplus_N M_2$. 
For $X \subseteq E(M_1) \cup E(M_2)$, let $X_i = \cl_i(X \cap E_i) \cup X$. 
It holds that 
$$\cl_M(X) = \cl_1(X_2 \cap E_1 ) \cup cl_2 (X_1 \cap E_2) \text{, and }$$
$$r_M(X) = r_{M_1}(X_2 \cap E_1) + r_{M_2}(X_1 \cap E_2 ) - r\big(T \cap (X_1 \cup X_2)\big).$$
\end{lemma} 

The operation of generalized parallel connection also commutes in the following sense.
\begin{lemma} 
\label{commute}
Let $K, M_1$ and $M_2$ be matroids such that $M_1 | T_1 = K | T_1$ and $M_2 | T_2 = K | T_2$. 
If $T_1$ is a modular semiflat in $M_1$ and $T_2$ is a modular semiflat in $M_2$, then 
$$M_2 \oplus_{N_2}(M_1 \oplus_{N_1} K) = M_1 \oplus_{N_1} (M_2 \oplus_{N_2} K).$$
\end{lemma}

\subsection{Amalgam width} 

Recall that the class of graphs of bounded tree-width can be introduced as the set of all subgraphs of a $k$-tree, where a $k$-tree is a graph that can be obtained by glueing two smaller $k$-trees along a clique of size $k$. 
Similarly, matroids of bounded branch-width can be introduced in terms of an operation taking two matroids of bounded branch-width and producing a larger matroid of bounded branch-width by glueing them along a low-rank separation.
The \a\ is also defined using a glueing operation. 
Analogously to the definition of tree-width, where some elements of the clique can be effectively removed after glueing takes place, the operation includes a set of elements to be deleted.
A typical situation when applying the glueing operation is illustrated on Figure \ref{f:glue}.


\begin{definition} 
\label{glueing}
Suppose we are given matroids $M_1, M_2,$ and $K$ such that $E(M_1) \cap E(M_2) \subseteq E(K)$. 
Furthermore, suppose we are also given a set $D \subseteq E(K)$.
Let $J_i := E(M_i) \cap E(K), i \in \{1, 2\}$ and assume the two conditions below hold: 
\begin{itemize} 
\item $M_i|J_i = K|J_i, i \in \onetwo$,
\item $J_1$ and $J_2$ are both modular semiflats in $K$.
\end{itemize} 
Then, the matroid $M_1 \amop M_2$ is defined as follows: 
$$M_1 \amop M_2 := \big((K \oplus_{J_1} M_1) \oplus_{J_2} M_2 \big) \setminus D.$$
We also say that the matroid $M_1 \amop M_2$ is a \textit{result of glueing of $M_1$ and $M_2$ along~$K$ and removing the elements $D$.} 
\end{definition} 

\begin{figure}[h]
 	\centering
    \executeiffilenewer{glueing.svg}{glueing.pdf}%
        {inkscape -z -D --file=glueing.svg --export-pdf=glueing.pdf --export-latex}%
    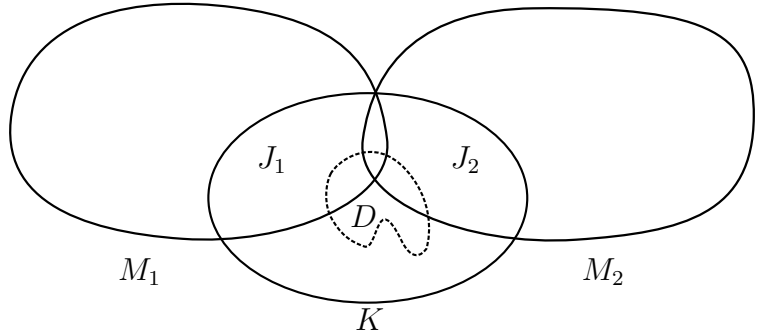
	\caption[]{
		$M_1, M_2$ are the matroids being combined, $K$ is a small matroid used to glue them together, and $D$ is a set of elements that are subsequently removed.}
	\label{f:glue}
\end{figure}
Note that Theorem \ref{connection} guarantees the matroid $M_1 \amop\ M_2$ to be well defined. 
We are now ready to introduce our width parameter. 

\begin{definition} 
\label{width} 
Matroid $M$ has \a\ at most $k \in \N$ if 
\begin{itemize}
\item $|E(M)| \le 1$, or 
\item there are matroids $M_1$ and $M_2$ of amalgam width at most $k$, a matroid $K$ satisfying $|E(K)| \le k$, and a choice of $D \subseteq E(K)$ such that $$M = M_1 \amop M_2.$$
\end{itemize} 
\end{definition} 

Note that the first condition can be weakened to $|E(M)| \le k$ without affecting the definition. 
Every finite matroid $M$ has an amalgam width at most $|E(M)|$. 
The \a\ of $M$ is the smallest $k$ such that $M$ has \a\ at most $k$.

The definition above naturally yields a tree-like representation of the construction of the matroid in question.

\begin{definition}
\label{decomposition}
Assume that $M$ is a matroid with amalgam width $k$.
Any rooted tree $\TREE$ satisfying either of the following statements is called an \defined{amalgam decomposition of $M$ of width at most $k$}: 
\begin{itemize} 
\item $|E(M)| \le 1$ and $\TREE$ is a trivial tree containing precisely one node, 
\item $M = M_1 \amop M_2$ and $\TREE$ has a root $r$ with children $r_1$ and $r_2$ such that the subtrees of $\TREE$ rooted at $r_1$ and $r_2$ are amalgam decompositions of $M_1$ and $M_2$ of width at most $k$.
\end{itemize} 
\end{definition} 

The above definition leads to a natural assignment of matroids to the nodes of $\TREE$: whenever a glueing operation is performed, we assign the resulting matroid to the node.
We use $M^\TREE(v)$ to refer to this matroid and say that the node $v$ represents $M^\TREE(v)$.
For an internal node $v \in \TREE$, we use $M_1^\TREE(v), M_2^\TREE(v), K^\TREE(v), D^\TREE(v), J^{\TREE}_1(v)$ and $J^{\TREE}_2(v)$ to denote the corresponding elements appearing in the glueing operation used to obtain $M^\TREE(v) = M_1^{\TREE}(v) {\oplus}^{K^{\TREE}(v), D^{\TREE}(v)} M_2^{\TREE}(v)$. 
If $v$ is a leaf of a decomposition $\TREE$, we let $M_1^\TREE(v)$ and $M_2^\TREE(v)$ be matroids with empty groundsets, $K^\TREE(v) := M^\TREE(v)$, and $D^\TREE(v) := \emptyset$.
Finally, we denote by $J^{\TREE}(v) \subseteq K(v)$ the set of elements used to glue $M(v)$ to its parent. 
More formally, we set $J^{\TREE}(v) := J^{\TREE}_i(u)$, where $i \in \onetwo$ is chosen depending on whether $v$ is a left or right child of $u$.
Since the decomposition under consideration is typically clear from context, we usually omit the upper index $\TREE$.

Strozecki \cite{S} introduces a similar parameter that uses the operation of a matroid $2$-sum instead of the generalized parallel connection.
However, its applicability is limited since it allows to join matroids only using separations of size at most 2 and thus a corresponding decomposition of a $3$-connected matroid $M$ has a width of $|E(M)|$.
The next proposition implies that the latter is able to express the $2$-sum operation as a special case. 
The parameter of Definition \ref{width} is therefore a generalization of the one from \cite{S}.

\begin{proposition} 
A $2$-sum of matroids $M_1$ and $M_2$ can be replaced by finitely many operations of generalized parallel connections and deletitions. 
\end{proposition} 


The proof of the above proposition is omitted. 

Next, we show that the amalgam width is a generalization of the branch-width parameter for finitely representable matroids in the sense that a bound on the value of branch-width implies a bound on the \a. 

\begin{proposition}
\label{btoa}
If $M$ is a matroid with branch-width $k$ and $M$ is representable over a finite field $\mathbf{F}$,
then the amalgam width of $M$ is at most $|\mathbf{F}|^{3k/2}$.
\end{proposition}

\begin{proof} 
Suppose we are given a (non-trivial) branch decomposition $\mathcal{B}$ of the matroid $M$ of width $k$, 
along with the representation of $M$ over $\mathbf{F}$.
We select an arbitrary internal node of the decomposition tree as its root node. 
Therefore, the elements of $M$ are vectors from $\mathbf{F}^d$ for some dimension $d \in \mathbf{N}$.
We construct an amalgam decomposition $\TREE$ of width at most $|\mathbf{F}|^{3k/2}$.
The leaves of $\TREE$ are the leaves of $\mathcal{B}$ and correspond to the same elements of $M$. 
Similarly, the internal nodes of $\TREE$ are the internal nodes of $\mathcal{B}$ and the associated matroids $K(v)$ (for $v \in \TREE$) are defined as follows.
Consider an internal node $v \in \mathcal{B}$ with children $v_1$ and $v_2$.
We use $E_1$ and $E_2$ to denote the set of elements of $M$ represented by the leaves in the subtree of $\mathcal{B}$ rooted at $v_1$ and $v_2$, respectively. 
We also let $E' := E(M) \setminus (E_1 \cup E_2)$.
Finally, we set $F \subseteq E(M)$ to be the set of all elements in at least two of the sets $\cl(E_1), \cl(E_2)$ and $\cl(E')$. 
Note that $\dim(F) \le \frac{3}{2} k$.
We construct $K(v)$ by taking as its ground set all linear combinations of vectors from $F$. 
Consequently, the sets $J_v^1$ and $J_v^2$ are $E(K(v)) \cap E_1$ and $E(K(v)) \cap E_2,$ respectively.

We need to check that the conditions of Theorem \ref{connection} are met. 
However, every flat $X$ in a matroid containing all $d$-dimensional vectors over $\mathbf{F}$ is modular, since for any flat $Y$ we have: 
$$r(X \cup Y) = \dim(X \cup Y) = \dim(X) + \dim(Y) - \dim(X \cap Y) = r(X) + r(Y) - r(X \cap Y),$$
where $\dim(\cdot)$ is the dimension of a vector subspace of $\mathbf{F}^d$.

The additional elements $E(K(v)) \setminus E(M)$ included in the construction above can be subsequently removed by including them in the set $D(u)$ at an ancestor $u$ of $v$, ensuring that the decomposition represents precisely the input matroid. 

\end{proof} 

\section{MSO properties}  
\label{msol} 

In this section, we show that the problem of deciding monadic second order properties is computationally tractable for matroids of bounded amalgam width. 
Our main theorem reads: 

\begin{theorem}
\label{mainmsollin}
MSO properties can be decided in linear time for matroids with amalgam width bounded by~$k$ (assuming the corresponding amalgam decomposition $\TREE$ of the matroid is given explicitly as a part of the input).
\end{theorem}


For the purpose of induction used in the proof of Theorem \ref{mainmsollin}, we need to slightly generalize the considered problem by introducing free variables.
The generated problem is given in Figure \ref{msolp}. 
\begin{figure}[h]
\centering
\begin{framed} 

\textbf{INPUT:} 
\begin{itemize}
\item an MSO formula $\psi$ with $p$ free variables,
\item amalgam decomposition $\TREE$ of a matroid $M$ with width at most $k$, 
\item a function $Q$ defined on the set $\{1, \ldots, p\}$ assigning the $i$-th free variable its value; specifically, $Q(i)$ is equal to an element of $E(M)$ if $x_i$ is an element variable, and it is a subset of $E(M)$ if $x_i$ is a set variable. 
\end{itemize}

\textbf{OUTPUT:} 
\begin{itemize}
\item \textit{ACCEPT} if $\psi$ is satisfied on $M$ with the values prescribed by $Q$ to the free variables of $\psi$.
\item \textit{REJECT} otherwise. 
\end{itemize}
\end{framed} 
\caption[]{The \msolp\ problem.}
\label{msolp}
\end{figure}
To simplify notation, let us assume that if $\psi$ is a formula with free variables, we use $x_i$ for the $i$-th variable if it appears in $\psi$ as an element variable and $X_i$ if it appears as a set variable. 

We prove the following generalization of Theorem \ref{mainmsollin}.

\begin{theorem}
\label{msollin}
The problem \msolp\ can be solved in linear time for matroids with amalgam width bounded by $k$ (assuming the corresponding amalgam decomposition $\TREE$ of the matroid is given as a part of the input).
\end{theorem}

Our aim in the proof of Theorem \ref{msollin} is to construct a linear time algorithm based on deterministic bottom-up tree automatons.  
Let us introduce such automatons.
\begin{definition} 
\defined{A finite tree automaton} is a $5$-tuple $(S, S_A, \delta, \Delta, \Sigma)$, 
where 
\begin{itemize}
\item $S$ is a finite set of states containing a special initial state $0$, 
\item $S_A \subseteq S$ is a non-empty set of accepting states, 
\item $\Sigma$ is a finite alphabet, 
\item $\delta: S \times \Sigma \rightarrow S$ is set of transition rules that determine a new state of the automaton based on its current state and the information, represented by $\Sigma$, contained in the current node of the processed tree, and 
\item $\Delta: S \times S \rightarrow S$ is a function combining the states of two children into a new state. 
\end{itemize}
\end{definition} 
Let us also establish the following simple notation.
\begin{definition}
Consider an instance of an \msolp\ problem. 
In particular, let $Q$ be the variable-assignment function as defined in Figure \ref{msolp}. 
For $F \subseteq E(M)$, we define the \defined{local view of $Q$ at $F$} to be the following function:  
\begin{equation*}
Q_F(i) := 
	\begin{cases} 
		Q(i) \cap F & \text{ if the $i$-th variable is a set variable, } \\ 
		Q(i)        & \text{ if the $i$-th variable is an element variable and $Q(i) \in F$, } \\
		\boxtimes   & \text{ otherwise},
	\end{cases} 
\end{equation*}
where $\boxtimes$ is a special symbol that is not an element of the input matroid. 
\end{definition}
The symbol $\boxtimes$ stands for values outside of $F$. 
We simplify the notation by writing $Q_v(x_i)$ instead of $Q_{E(K(v))}(i)$, where $v$ is a node of an amalgam decomposition~$\TREE$.

The alphabet $\Sigma$ of the automaton we construct will correspond to the set of all possible ``configurations'' at a node $v$ in an amalgam decomposition of width at most $k$. 
A finite tree automaton processes a tree (in our case $\TREE$) from its leaves to the root, assigning states to each node based on the information read in the node and on the states of its children. 
When processing a node whose two children were already processed the automaton calculates the state $s := \Delta(s_1, s_2)$, where $s_1$ and $s_2$ are the states of the children, and moves to the state $\delta(s, q)$, where $q \in \Sigma$ represents the information contained in that node of the tree. 
If the state eventually assigned to the root of the tree is contained in the set $S_A$, we say that the automaton accepts.
It rejects otherwise. 

As a final step of our preparation for the proof of Theorem \ref{msollin}, we slightly alter the definition of an MSO formula by replacing the use of $\ind(X)$ predicate with the use of $x_1 \in \cl(X_2)$, where $\cl(\cdot)$ is the closure function of $M$.
The predicate $\ind(X)$ can be expressed while adhering to the altered definition as follows: 
\begin{align*} 
\ind(X) & \equiv \neg \big( \exists e \in X: \cl(X) = \cl(X \setminus \{e\}) \big). 
\end{align*} 

\begin{proof}[of Theorem \ref{msollin}] 
We proceed by induction on the complexity of the formula $\psi$, starting with simple formulas such as $x_1 = x_2$ or $x_1 \in X_2$. 
In each step of the induction, we design a tree automaton processing the amalgam decomposition tree $\TREE$ and correctly solving the corresponding \msolp\ problem. 
As already mentioned, the alphabet will encode all possible non-isomorphic choices of the matroid $K(v)$, sets $\Jv, \JvI, \JvII,$ and $D(v)$ combined with all possible local views of $Q$ at $v$, allowing this information to be read when processing the corresponding node. 
Note that if $k$ is bounded, the size of the set $\Sigma$ of such configurations is bounded. 
Since the automaton size does not depend on $n$ and the amount of information read in each node of $\TREE$ is bounded by a constant (assuming bounded \a), we will be able to conclude that the running time of our algorithm, which will just simulate the tree automaton, is linear in the size of $\TREE$. 
To start the induction, we first consider the case $\psi = x_1 \in X_2$.
Such instances of \msolp\ can be solved by the automaton given in Figure $\ref{aut1}$. 
This automaton stays in its original state if $x_1$ is assigned $\boxtimes$ by the local view of $Q$ at $E(K(v))$.
Otherwise, it moves to designated ACCEPT and REJECT states based on whether $\Qv(x_1) \in \Qv(X_2)$ holds.
The set $S_A$ is defined to be $\{\ACC\}$.
The function $\Delta: S \times S \rightarrow S$ assigns the ACCEPT state to any tuple containing an ACCEPT state. 
Similarly for the REJECT state.
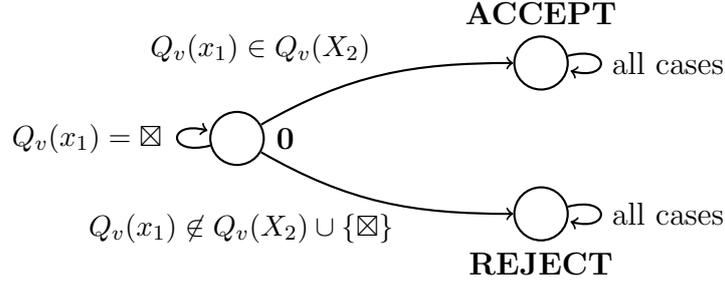
\begin{figure}[h]
\begin{center}\begin{tikzpicture}[style=thick]

	\node[vertex] (st0) at (0, 0)  [label=right:$\mathbf{0}$]    {};
	\node[vertex] (stA) at (4, 1)  [label=above:\textbf{ACCEPT}] {};
	\node[vertex] (stR) at (4, -1) [label=below:\textbf{REJECT}] {};

	\path[oedge] (st0) edge [out=30, in=180]  node[above=0.35cm,left=0cm] {{\small $Q_v(x_1) \in Q_v(X_2)$}} (stA);
	\path[oedge] (st0) edge [out=-30, in=180] node[above=-0.35cm,left=-0.3cm] {{\small $Q_v(x_1) \not\in Q_v(X_2) \cup \{\boxtimes\}$}} (stR);
	\path[oedge] (st0) edge [loop left] node[right=-2.33cm] {{\small $Q_v(x_1) = \boxtimes$}} (st0);
	\path[oedge] (stA) edge [loop right] node[right] {all cases} (stA);
	\path[oedge] (stR) edge [loop right] node[right] {all cases} (stR);

\end{tikzpicture}\end{center}
\caption[]{
	The states and transition rules $\delta$ of the tree automaton for the formula $x_1~\in~X_2$.
	Here, $v$ is the currently processed node of the amalgam decomposition. 
	The names of the states are typed using bold font.
}
\label{aut1}
\end{figure}
We are guaranteed not to encounter the situation where one child node is in the ACCEPT state and the other in the REJECT state, since the free variable assignment function $Q$ maps $x_1$ precisely to one element of $E(M)$.
It is clear that this tree automaton correctly propagates the information of whether $x_1 \in X_2$ or not from the leaf representing the value of $x_1$ to the root of $\TREE$. 


The cases of formulas $x_1 = x_2$ and $X_1 = X_2$ can be handled similarly.
For formulas of the form $\psi_1~\lor~\psi_2$, we construct the automaton by taking the Cartesian product of the automatons $A_1 = (S^1, S_A^1, \delta^1, \Delta^1, \Sigma^1)$ and $A_2 = (S^2, S_A^2, \delta^2, \Delta^2, \Sigma^2)$ for the formulas $\psi_1$ and $\psi_2$, respectively.
Specifically, 
\begin{align*} 
\Sigma & = \Sigma^1 \times \Sigma^2, \\ 
S & = S^1 \times S^2, \\ 
S_A & = (S_A^1 \times S^2) \cup (S^1 \times S_A^2), \\
\Delta\big((x, y)\big) & = \big(\Delta^1(x), \Delta^2(y)\big), \\
\delta\big((x, y), (q, r)\big) & = \big(\delta^1(x, q), \delta^2(y, r)\big).
\end{align*} 
Informally, the two automatons run in parallel and the new automaton accepts precisely if at least one of the two is in an accepting state. 
A formula of the form $\neg \psi$ can be processed by the same automaton as $\psi$, except we change the set accepting states to their complement. 

The connectives $\land, \Rightarrow, \ldots$ can be expressed using $\lor$ and $\neg$ by a standard reduction.

So far, we did not apply most of the properties of amalgam decompositions, including Lemma \ref{closureissimple}, which constrains the possible ways in which closures of sets can behave in a matroid resulting from a generalized parallel connection.
This comes into play when constructing an automaton for the formula $x_1 \in \cl(X_2)$.


Let us first give an informal description. 
When processing a node $v$ of $\TREE$, we are able to see the elements of $K(v)$, to query the independent sets on $E(K(v))$, and to see the local view of $Q(X_2)$ at $E(K(v))$.
Our strategy will be to compute $\cl_{M}(X_2)$ restricted to $E(K(v))$ and determine whether $x_1$ is contained in it.
However, the state at $v$ does not encode necessary information about the remaining part of $M$, i.e., the part represented by the nodes of $\TREE$ that are not descendants of $v$.
The matroid $M(v)$ is joined to this part by a generalized parallel connection using the modular flat $J(v)$. 
Lemma~\ref{closureissimple} says that the remaining part of $M$ can influence the restriction of the closure of $X_2$ on $E(K(v))$ only through forcing some of the elements of this modular flat into the closure. 
Since $|J(v)|$ is bounded, we can precompute the behavior of the resulting closure for all possible cases. 
This information is encoded in the state of the finite automaton passed to the parent node. 
The parent node can then use the information encoded in the states corresponding to its children when precomputing its intersection with $\cl_{M}(X_2)$. 
We formalize this approach using the following definition. 
\begin{definition} 
\label{typed}
Let $v$ be a node of an amalgam decomposition $\TREE$ of $M$ and $X$ be a subset of $E(M)$. 
A map $f_v^X$ from $2^{\Jv} \rightarrow 2^{\Jv}$ satisfying $$f_v^X(Y) = \cl_{M(v)}\big(\big(X \cap E(M(v))\big) \cup Y\big) \cap \Jv$$ is called \defined{the type of a node $v$ with respect to $X$}.
\end{definition} 

When processing a node $v$, we can assume we are given the types $f_1^X$ and $f_2^X$ of the children of $v$ 
and we want to determine the type of $v$. 
The type is then encoded into the state of the finite automaton (along with the information for which choices of $Y \subseteq \Jv$ the formula $\psi$ holds) and is passed to the parent node. 
This information is then reused to determine the type of the parent node, etc. 
This process is captured by the following definition.

\begin{definition} 
Let $v$ be a node of an amalgam decomposition $\TREE$ of a matroid $M$, $v_1$ and $v_2$ the children of $v$, and $X$ a subset of $E(M)$.
If $f_{v_1}^X$ is the type of $v_1$ with respect to $X$ and $f_{v_2}^X$ is a type of $v_2$ with respect to $X$, 
we say that the type $f_{v_1}^X +_{K(v)} f_{v_2}^X$ of $v$ is \defined{the join of $f_{v_1}^X$ and $f_{v_2}^X$} if for every subset $Y$ of $\Jv$ it holds that $f_{v_1}^X +_{K(v)} f_{v_2}^X = Z \cap \Jv$, where $Z$ is the smallest subset of $E(K(v))$ such that
\begin{itemize} 
\item $f_{v_1}^X(Z \cap \JvI) = Z \cap \JvI,$
\item $f_{v_2}^X(Z \cap \JvII) = Z \cap \JvII,$ 
\item $Z \supseteq Y \cup (X \cap E(K(v))).$
\end{itemize} 
\end{definition} 

Lemma~\ref{closureissimple} implies that $f_{v_1}^X +_{K(v)} f_{v_2}^X$ is the type of the node $v$ with respect to $X$.
Observe that the type $f_1^X +_{K(v)} f_2^X$ in the above definition is determined by $f^X_{v_1}, f^X_{v_2}, K(v)$ and $X \cap E(K(v))$ -- each of which has bounded size. 
This implies that the computation of the type $f_1^X +_{K(v)} f_2^X$ can be wired in the transition function of the automaton. 
Deciding if $Q(x_1) \in \cl(X_2) \cap \Jv$ is then reduced to verifying if $Q(x_1) \in f_v^{X_2}(Y)$ for a particular choice of $Y$.

The case of a formula $\exists x : \psi$ is solved by a standard argument of taking the finite tree automaton recognizing $\psi$ and transforming it to a non-deterministic tree automaton that tries to guess the value of $x$ (in our case, the automaton also checks if this guessed value of $x$ lies in the set $D(v)$ of deleted elements).
This non-deterministic tree automaton has a finite number of states by induction. 
A non-deterministic finite tree automaton can be simulated using a deterministic finite tree automaton with up to an exponential blow-up of the number of states,
leading to the conclusion that a formula of this form can again be decided by a deterministic finite tree automaton. 
The case $\exists X : \psi$ is solved analogously. 


Since the algorithm simulating the automaton on $\TREE$ spends $\O(1)$ time in each of the nodes of $\TREE$, there exists a linear time algorithm solving the problem from the statement of the theorem. 
\end{proof}

\section{Tutte polynomial} 

The Tutte polynomial is an important combinatorial invariant defined for graphs and matroids. 
Values of the polynomial encode, e.g., the number of its bases.
In the case of graphs, the values of the polynomial also give numbers of $k$-colorings. 
\begin{definition} 
Let $M$ be a matroid with a ground set $E$. 
The Tutte polynomial of $M$ is a bivariate polynomial $$T_M(x, y) := \sum_{F \subseteq E} (x - 1)^{r(E) - r(F)} (y - 1)^{|F|-r(F)}.$$
\end{definition} 
The main result of this section is the following theorem: 
\begin{theorem} 
\label{tuttelin}
For every $k \in N$, there exists a polynomial-time algorithm that given an amalgam decomposition with width at most $k$ of a matroid $M$ computes the coefficients of the Tutte polynomial of $M$ (assuming the corresponding amalgam decomposition $\TREE$ of the matroid is given explicitly as a part of the input).
The degree of the polynomial in the running time estimate of the algorithm is independent of $k$.
\end{theorem} 

Before we start the proof of Theorem \ref{tuttelin}, we introduce a slight modification of the notion of amalgam decompositions.
\begin{definition} 
An amalgam decomposition $T$ is nice if for each node $v$ the sets $\JvI$ and $\JvII$ are disjoint.
\end{definition} 
Every amalgam decomposition can be transformed into a nice amalgam decomposition such that the width increases only by a constant factor by duplicating the elements of $J(u) \cap \Jv$ and subsequently deleting the duplicates by including them in the set $D_{w}$ for some ancestor $w$ of $u, v$.
\begin{lemma} 
\label{nice} 
Let $T$ be an amalgam decomposition of matroid $M$ with width $k$. 
Then there exists a nice amalgam decomposition of $M$ with width at most $2k$. 
Moreover, such decomposition can be found in linear time.
\end{lemma} 


We are now ready to prove Theorem \ref{tuttelin}.
\begin{proof}[of Theorem \ref{tuttelin}] 
The key idea is to count the number of sets with given size, rank, and type (see Definition \ref{typed}). 
This way, we get the coefficients in the definition of the polynomial.
For a given set $X \subseteq E(M(v))$, the type of $v$ allows us to compute the closure of $X$ on $E(K(v))$ without any additional knowledge of the structure of $M(v)$. 
In the proof of Theorem \ref{msollin}, we have seen that we can compute the type of a node $v$ with respect to a set $X$ from the types of its children in constant time. 
The number of subsets of $E(M(v))$ with given rank $r$, size $s$ and type $f$ is denoted by $count_v(r,s,f)$. 
Algorithm $1$ computes these numbers for a node $v$ using the numbers computed for its children.
\begin{algorithm}
\label{tuttealg} 
INPUT: vertex $v \in \TREE$ with children $v_1, v_2$ \\
OUTPUT: $\mc_v(r, s, f)$ for all $\text{ranks } r, \text{ set sizes } s \text{ and set types } f$ \\
\vspace{2mm}
initialize $\mc_v(r, s, f) \leftarrow 0$ for all $r, s, f$. \\
\For{$s_1 \in {1, 2, \ldots, |E(M(v_1))|}$} { 
\For{$s_2 \in {1, 2, \ldots, |E(M(v_2))|}$} {
\For{$r_1 \in {1, 2, \ldots, r(M(v_1))}$} {
\For{$r_2 \in {1, 2, \ldots, r(M(v_2))}$} { 
\For{$f_1$ type at $v_1$} {
\For{$f_2$ type at $v_2$} {
$f \leftarrow f_1 +_{K(v)} f_2$ \\ 
$s \leftarrow s_1 + s_2$ \\
$\mc_v(r, s, f) \leftarrow \mc_v(r, s, f) + \mc_{v_1}(r_1, s_1, f_1) \times \mc_{v_2}(r_2, s_2, f_2)$
}
}
}
}
}
}
\caption{Computing the $\mc_v$ function for a branching node $v \in \TREE$.} 
\end{algorithm} 

The algorithm for Theorem \ref{tuttelin} first computes the value of $\mc_v$ for leaves of $\TREE$, then picks an arbitrary node such that the value of $\mc_v$ for both of its children were already determined and applies Algorithm \ref{tuttealg}. 
The computation for leaves is trivial.
There are only two possible cases: either the element represented by the leaf is a loop or is not. 
At the root $r$, there is only one type $f_0$ (since $J_r$ is empty) and the number $\mc_v(r, s, f_0)$ is the number of sets of $M$ of a given rank and size.

Let us turn our attention to the analysis of the time complexity of this algorithm.
Each of the two outer loops of Algorithm \ref{tuttealg} makes at most $n$ iterations. 
The loops iterating over the rank make at most $r(M)$ iteration each. 
The number of types can be bounded by a function of $k$ and is therefore constant with respect to $n$. 
Thus, the number of iterations of the remaining two loops is again also constant. 
We conclude that the total time complexity of Algorithm \ref{tuttealg} is $\O\big(n^2 r^2\big)$, where $r := r(M)$.
For the computation of the Tutte polynomial, we need to call Algorithm \ref{tuttealg} for each branching node. 
So, the resulting time complexity of our algorithm is $\O\big(n^3 r^2\big)$.
\end{proof}
\label{con}


\section{Conclusion} 

Both the Theorem \ref{msollin} of Section \ref{msol} and the Theorem \ref{tuttelin} assume that the amalgam decomposition is given as a part of the input. 
This assumption can be removed for matroids that are representable over a fixed finite field, since the proof of Proposition \ref{btoa} gives a linear time algorithm constructing an amalgam decomposition from a branch decomposition. 
Therefore, we can use a polynomial-time algorithm \cite{O1, Hcon} for constructing a branch decomposition and then convert it to an amalgam decomposition of width bounded by a constant multiple of the original branch-width.
Similarly, it can be shown that the branch decomposition of a representable matroid can be obtained from an amalgam decomposition in a natural way. 
However, we have not been able to settle the complexity of constructing amalgam decomposition of (approximately) optimal width for a general oracle-given matroid.

\section{Acknowledgment} 
The authors would like to thank Dan Kr\' al' for his valuable insights into the problem.

%
%






\bibliographystyle{siam}
\bibliography{all}

\end{document}

%% file: glueing.pdf_tex
\begingroup%
  \makeatletter%
  \providecommand\color[2][]{%
    \errmessage{(Inkscape) Color is used for the text in Inkscape, but the package 'color.sty' is not loaded}%
    \renewcommand\color[2][]{}%
  }%
  \providecommand\transparent[1]{%
    \errmessage{(Inkscape) Transparency is used (non-zero) for the text in Inkscape, but the package 'transparent.sty' is not loaded}%
    \renewcommand\transparent[1]{}%
  }%
  \providecommand\rotatebox[2]{#2}%
  \ifx\svgwidth\undefined%
    \setlength{\unitlength}{278.15878957bp}%
    \ifx\svgscale\undefined%
      \relax%
    \else%
      \setlength{\unitlength}{\unitlength * \real{\svgscale}}%
    \fi%
  \else%
    \setlength{\unitlength}{\svgwidth}%
  \fi%
  \global\let\svgwidth\undefined%
  \global\let\svgscale\undefined%
  \makeatother%
  \begin{picture}(1,0.44619252)%
    \put(0,0){\includegraphics[width=\unitlength]{glueing.pdf}}%
    \put(0.14639737,0.07177973){\color[rgb]{0,0,0}\makebox(0,0)[lb]{\smash{$M_1$}}}%
    \put(0.76798846,0.07176612){\color[rgb]{0,0,0}\makebox(0,0)[lb]{\smash{$M_2$}}}%
    \put(0.46394894,0.00380431){\color[rgb]{0,0,0}\makebox(0,0)[lb]{\smash{$K$}}}%
    \put(0.3320229,0.22317483){\color[rgb]{0,0,0}\makebox(0,0)[lb]{\smash{$J_1$}}}%
    \put(0.59195797,0.22316127){\color[rgb]{0,0,0}\makebox(0,0)[lb]{\smash{$J_2$}}}%
    \put(0.45741216,0.14376857){\color[rgb]{0,0,0}\makebox(0,0)[lb]{\smash{$D$}}}%
  \end{picture}%
\endgroup%

%% file: width.bbl
\begin{thebibliography}{99}
\bibitem{BX} 
B.-M. Bui-Xuan, J. A. Telle, M. Vatshelle: 
Boolean-width of graphs. 
4th International Workshop on Parameterized and Exact Computation (IWPEC 2009). LNCS, vol. 5917, Springer, Heidelberg, 2009, 61-74.
\bibitem{C1} 
B. Courcelle: 
The monadic second-order logic of graph I. 
Recognizable sets of finite graphs. 
Inform. and Comput. 85, 1990, 12–75.
\bibitem{C2}
B. Courcelle: 
The expression of graph properties and graph transformations in monadic second-order logic.
G. Rozenberg (ed.), Handbook of graph grammars and computing by graph transformations, Vol. 1: Foundations, World Scientific, 1997, 313–400.
\bibitem{C3}
B. Courcelle, S. Olariu:
Upper bounds to the clique width of graphs.
Discrete Applied Mathematics, Volume 101, Issues 1–3, 2000, 77-114.
\bibitem{CMR}
B. Courcelle, J. A. Makowsky, U. Rotics. 
On the fixed parameter complexity of graph enumeration problems definable in monadic second-order logic. 
Discrete Applied Mathematics, 108(1-2):23–52, 2001.
\bibitem{Htw}
P. Hlin\v en\'y, G. Whittle:
Matroid tree-width. 
Technical report, Technical University Ostrava, Ostrava, Czech Republic, 2003.
\bibitem{Hcon}
P. Hlin\v en\'y: 
A parametrized algorithm for matroid branch-width, 
SIAM J. Computing 35, 2005, 259–277.
\bibitem{Hbw} 
P. Hlin\v en\'y:
Branch-width, parse trees, and monadic second-order logic for matroids.
J. Combin. Theory Ser. B 96, 2006, 325–351.
\bibitem{Hnon}
P. Hlin\v en\'y: 
On matroid representatibility and minor problems.
Proc. of MFCS 2006, LNCS vol. 4192, Springer, Berlin, 2006, 505–516.
\bibitem{HOcon}
P. Hlin\v en\'y, S. Oum: 
Finding branch-decomposition and rank-decomposition. 
SIAM J. Computing 38, 2008, 1012–1032.
\bibitem{GKO} 
T. Gaven\v ciak, D. Kr\'al'k, S. Oum: 
Deciding firost order logic properties of matroids.
ICALP, 2012, 239-250.
\bibitem{Crap} 
D. Kr\'al': 
Representatons of matroids of bounded branch-width.
STACS, 2007, 224-235.
\bibitem{K} 
D. Kr\'al': 
Decomposition width of matroids. 
Discrete Applied Mathematics 160(6), 2012, 913-923.
\bibitem{O1} 
S. Oum, P. D. Seymour: 
Approximating clique-width and branch-width. 
J. Combin. Theory Ser. B 96, 2006, 514–528.
\bibitem{Ocert}
S. Oum, P. D. Seymour: 
Certifying large branch-width.
Proc. of SODA 2006, SIAM, 2006, 810–813.  
\bibitem{Oapprox} 
S. Oum, P. D. Seymour: 
Approximating clique-width and branch-width, 
J. Combin. Theory, Ser. B., 2006, 514-528.
\bibitem{O2}
S. Oum, P. D. Seymour: 
Testing branch-width. 
J. Combin. Theory Ser. B 97, 2007, 385–393.
\bibitem{O}
J. Oxley: 
Matroid Theory, Second Edition. 
Oxford University Press (2011). 
\bibitem{RS}
N. S. Robertson, P. D. N. Seymour:
Graph minors. X. Obstructions to tree-decomposition.
Journal of Combinatorial Theory 52 (2): 153–190 (1991).
\bibitem{Sey} 
P. Seymour: 
Recognizing graphic matroids.
Combinatorica 1 (1981), 75–78.
\bibitem{S}
Y. Strozecki:
Monadic second-order model-checking on decomposable matroids. 
Discrete Applied Mathematics (2011).
\end{thebibliography}
